\documentclass[11pt]{article}
\usepackage{theorem}
\usepackage{epsfig}
\usepackage{amsmath}
\usepackage{amssymb}
\usepackage{amsfonts}
\usepackage{pst-node}
\usepackage{pst-plot}
\usepackage{wrapfig}

\usepackage{vmargin}
\setmarginsrb{1.3in}{1.2in}{1.3in}{1.1in}{0mm}{0mm}{10pt}{15mm}

\newcommand{\bvec}[1]{\ensuremath{\mathbf{#1}}}
\newtheorem{thm}{Theorem}

\newtheorem{lemma}[thm]{Lemma}

\newtheorem{cor}[thm]{Corollary}
\theorembodyfont{\rmfamily}

\title{Maximal f-vectors of Minkowski sums of large numbers of polytopes}

\date{}
\author{
 Christophe Weibel\footnote{Department of Mathematics and Statistics, McGill, Montr\'eal, Canada. \texttt{weibel@math.mcgill.ca}}
}

\newcommand{\qed}{\mbox{}\hspace*{\fill}\nolinebreak\mbox{$\rule{0.6em}{0.6em}
$}}
\newenvironment{proof}{\begin{rm}\par\smallskip\noindent {\bf Proof.}\quad}{$\qed$\end{rm}\par\medskip}

\begin{document}
\thispagestyle{empty}
\maketitle

\begin{abstract}
  It is known that in the Minkowski sum of $r$ polytopes in dimension
  $d$, with $r<d$, the number of vertices of the sum can potentially
  be as high as the product of the number of vertices in each
  summand~\cite{Fukuda07}. However, the number of vertices for sums of
  more polytopes was unknown so far.

  In this paper, we study sums of polytopes in general orientations,
  and show a linear relation between the number of faces of a sum of
  $r$ polytopes in dimension $d$, with $r\geq d$, and the number of
  faces in the sums of less than $d$ of the summand polytopes. We
  deduce from this exact formula a tight bound on the maximum possible
  number of vertices of the Minkowski sum of any number of polytopes
  in any dimension. In particular, the linear relation implies that a
  sum of $r$ polytopes in dimension $d$ has a number of vertices in
  $O(n^{d-1})$ of the total number of vertices in the summands, even
  when $r\geq d$. This bound is tight, in the sense that some sums do
  have that many vertices.
\end{abstract}
\setcounter{page}{0}
\clearpage
\section{Introduction}
The \emph{Minkowski sum} of two polytopes is defined as
$P_1+P_2=\{x_1+x_2\;:\;x_1\in P_1,\;x_2\in P_2\}$. Minkowski sums are
of interest in various fields of theoretical and applied
mathematics. While some applications only require sums of two
polytopes in low dimensions (e.g. motion planning~\cite{Halperin04}\cite{Lozano79}),
others use iterative sums of many polytopes in higher
dimensions~\cite{Pachter05}\cite{Zhang09}. It is therefore desirable to
study the complexity of such
sums. 

A trivial bound on the number of vertices of a sum is found as
follows. Every vertex of a Minkowski sum decomposes into a sum
of vertices of the summands. Therefore, there cannot be more vertices
in the sum than there are possible decompositions. Thus, a trivial
bound on the number of vertices in a Minkowski sum is the product of
the number of vertices in the summands. That is, if $P_1,\ldots,P_r$
are polytopes, and $f_0(P)$ is the number of vertices of a polytope,
then
$f_0(P_1+\cdots+P_r) \leq f_0(P_1)\cdots f_0(P_r)$.

If we sum $r$ polytopes in dimension $d$, with $r<d$, then the trivial
bound is tight, that is, it is possible to choose summands with any
number of vertices so that their sum has as many vertices as the trivial
bound~\cite{Fukuda07}.

However, if we sum $r$ polytopes in dimension $d$ with $r\geq d$, the
trivial bound cannot be reached, except when summing $d$
segments~\cite{Sanyal07}\cite{Weibel07}. We assume here and in the
rest of the article that all polytopes have at least two vertices,
since a summand of only one vertex can be ignored without changing the
properties of the sum.

The f-vector of a polytope encodes its number of faces of different
dimensions. Maximal f-vectors are obtained for a particular class of
Minkowski sums, called sums of polytopes \emph{in general
  orientations}. We will therefore restrict our study to such sums.

We recently presented in~\cite{Fogel09} a result on sums of
$3$-dimensional polytopes in general orientations. We showed that the
number of vertices in a sum of $r$ summands can be deduced from the
number of vertices in the summands and the number of vertices in sums
of each of the ${r\choose 2}$ pairs of summands. Using this result, we
found a tight upper bound on the number of vertices and facets in sums
of $3$-dimensional polytopes.

The basic reasoning of this previous result is to define a unique
witness, called \emph{western-most corner}, for all but two vertices of
a polytope. These witnesses have the property that a western-most corner
for a Minkowski sum of any number of summands is also necessarily a
western-most corner for the sum of some pair of the summands. So the
number of western-most corners in the total sum, and thus its number
of vertices, can be found by examining sums of one or two summands only.

This prompted us to examine the possibility of extending the reasoning
to higher dimensions and other faces, which resulted in this
article. Our main result is presented in Theorem \ref{thm:main}. It is
a linear relation between the number of faces of a sum of $r$
polytopes and the number of faces in the sums of less than $d$ of the
summand polytopes:
\begin{thm}\label{thm:main}
  Let $P_1,\ldots,P_r$ be $d$-dimensional polytopes in general
  orientations, $r\geq d$, and each polytope full-dimensional. For any
  $k$ in $0,\ldots,d-1$,
$$
f_k(P_1+\cdots+P_r)-\alpha=\sum_{j=1}^{d-1}(-1)^{d-1-j} {{r-1-j} \choose {d-1-j}} \sum_{S\in \mathcal{C}_j^r}(f_k(P_S)-\alpha),
$$
where $\mathcal{C}_j^r$ is the family of subsets of $\{1,\ldots,r\}$ of
cardinality $j$, $P_S$ is the sum of polytopes $\sum_{i\in S}P_i$;
 $\alpha=2$ if $k=0$ and $d$ is odd,  $\alpha=0$ otherwise.
\end{thm}
A slightly more general result also applies when summands are not
full-dimensional. The intuitive explanation of the theorem is that for
any face of the whole sum, we can find a witness of its existence by
examining the faces of the same dimension in sums of $d-1$
summands. However, if that witness exists in some sum of $d-2$
summands, we will find it in many different sums of $d-1$ summands. So
we need to offset this by removing an appropriate number of times the
witnesses in sums of $d-2$ summands. But that in turn removes too many
times witnesses that exist in some sum of $d-3$ summands, so we need
to add them back, etc. This implies  that the total sum is smaller
than the term for $j=d-1$:
\begin{cor}\label{cor}
  Let $P_1,\ldots,P_r$ be $d$-dimensional polytopes in general
  orientations, $r\geq d$, and each polytope full-dimensional. For any
  $k$ in $0,\ldots,d-1$,
$$
f_k(P_1+\cdots+P_r)\leq \sum_{S\in \mathcal{C}_{d-1}^r}f_k(P_S).
$$
\end{cor}
From this result, we deduce bounds for the maximum possible
number of vertices in a Minkowski sum of polytopes, for fixed number of
vertices in the summands.

We find in particular that a sum of $r$ polytopes in dimension $d$,
$r\geq d$, where summands have $n$ vertices in total, has less than ${n
  \choose {d-1}}$ vertices, which is in $O(n^{d-1})$. In the case
where each summand has at most $n$ vertices, then the number of
vertices of the sum is less than ${r \choose {d-1}} n^{d-1}$, which is
in $O(r^{d-1}n^{d-1})$. This is better than the previous known bound
which was in $O(r^{d-1}n^{2(d-1)})$~\cite{Gritzmann93}.


In the rest of the article, we shortly present the theory in
Section~\ref{sec:mink}. We first give an introduction to the concepts
of west and western-most corner in three dimensions in
Section~\ref{sec:3d}, then extend them formally to higher dimensions
in Section~\ref{sec:wit}. We examine in Section~\ref{sec:max} what are
the maximum possible number of faces of a Minkowski sum and conclude
in Section~\ref{sec:con}.

\section{Minkowski sums}\label{sec:mink}
Let $P_1,\ldots,P_r$ be given polytopes. Their Minkowski sum is the polytope
defined as $P_1+\cdots+P_r = \{x_1+\cdots+x_r\::\:x_i\in P_i,\;\forall i\}$.
We assume in the following, and in the rest of the article, that
every polytope is full-dimensional.

A nontrivial \emph{face} of a polytope $P$ in dimension $d$ is the
intersection of $P$ with a support hyperplane of $P$. Vertices, edges,
facets, ridges are the faces of dimension $0$, $1$, $d-1$ and $d-2$
respectively. Thus, we can associate to each vector in the unit sphere
$S^{d-1}$ a face of the polytope, which is the intersection of the polytope
with the support hyperplane to which the vector is outwardly normal:
$\mathcal{S}(P;l) = \{x\in P\;:\:\langle l,x\rangle \geq \langle l,y\rangle,\;\forall y\in P\}$.

Conversely, each face $F$ of a $d$-dimensional polytope $P$ can be
associated with a region of the sphere $S^{d-1}$, called the
\emph{normal region}, which is the set of unit vectors outwardly
normal to some support hyperplane of $P$ whose intersection with $P$
is $F$: $\mathcal{N}(F;P)=\{l\in
S^{d-1}\;|\,F=\mathcal{S}(P;l)\}=\{l\in S^{d-1}\;|\;\langle
l,x\rangle>\langle l,y\rangle,\;\forall x\in F,\;y\in P\setminus
F\}$. The normal region of a facet of $P$ is thus a single point of
$S^{d-1}$, corresponding to the unit vector outwardly normal to the
facet. The normal region of a face of dimension $k$ is a relatively
open subset of $S^{d-1}$ of dimension $d-1-k$.

We call a subset of the sphere $S^{d-1}$ \emph{spherically convex} if
for any two points in the subset, any shortest arc of great circle
between the two points is inside the subset.\footnote{There exist
  different definitions of convexity on a sphere. Note that according
  to this one, the only convex set containing antipodal points is the
  whole sphere.} If the polytope $P$ is full-dimensional, the normal
regions of faces of $P$ are all disjoint, relatively open and
spherically convex. They determine a subdivision of $S^{d-1}$ into a
spherical cell complex which we call the \emph{Gaussian map} of the polytope:
$\mathcal{G}(P)=\{\mathcal{N}(F;P)\;:\;F\mbox{ face of }P\}$.

A property of Minkowski sums is that faces of the sum have a unique
\emph{decomposition} in faces of the summand. Let $F$ be a face of the
Minkowski sum $P=P_1+\cdots+P_r$, and $l$ be in $\mathcal{N}(F;P)$.
Then $F=F_1+\cdots+F_r$, where $F_i= \mathcal{S}(P_i;l)$ is a face of
$P_i$. We can deduce that the normal region of a face of the sum is
equal to the intersection of the normal regions of the faces in its
decomposition:
$\mathcal{N}(F;P)=\mathcal{N}(F_1;P_1)\cap\cdots\cap\mathcal{N}(F_r;P_r)$.
Thus the Gaussian map of the Minkowski sum is the \emph{common
  refinement} of the Gaussian map of the summands:
$$
\mathcal{G}(P_1+\cdots+P_r)=
\{\mathcal{N}(F_1;P_1)\cap\cdots\cap\mathcal{N}(F_r;P_r)\;:\;
F_i\mbox{ face of }P_i\}.
$$
A polytope and its Gaussian map being dual structures, it is possible
to study the number of faces of a polytope by studying the number of
cells of its Gaussian map.

We say that a face of a Minkowski sum has an \emph{exact
  decomposition} $F=F_1+\cdots+F_r$ when its dimension is the sum of
the dimension of the faces in its decomposition:
$\dim(F)=\dim(F_1)+\cdots+\dim(F_r)$. That is, the decomposition is
exact when there are no two parallel segments inside different faces
in the decomposition. We say that polytopes are \emph{in general
  orientations} when all faces of their Minkowski sum have an exact
decomposition.


For fixed f-vectors of summands, the maximum number of faces of any
dimension in the sum can always be reached by summands in general
orientations~\cite{Fukuda08}. Therefore, we can assume summands are in
general orientations when looking for sums with maximum number of
faces.

Let $F$ be a face of the Minkowski sum $P=P_1+\cdots+P_r$ of
$d$-dimensional polytopes in general orientations. The face $F$
decomposes into a sum $F_1+\cdots+F_r$ of faces of the summands, with
$\dim(F)=\dim(F_1)+\cdots+\dim(F_r)$. Even if $r\geq d$, there are
therefore at most $\dim(F)$ faces in the decomposition that have a
dimension of more than $0$. Let the \emph{support} $I_F\subseteq
\{1,\ldots,r\}$ of $F$ be the set of indices of these faces, with
$|I_F|\leq \dim(F)$. Note that for any subface $G$ of $F$, $G$
decomposes into a sum $G_1+\cdots+G_r$, where $G_i\subseteq F_i$ for
all $i$; and so, $I_G\subseteq I_F$.

For any $S=\{i_1,\ldots,i_s\}$ subset of $\{1,\ldots,r\}$, let us
define the \emph{partial sum} $P_S=P_{i_1}+\cdots+P_{i_s}$.
\begin{lemma}
  Let $F$ be a facet of a Minkowski sum $P=P_1+\cdots+P_r$ of
  $d$-dimensional polytopes in general orientations. Its normal region
  $\mathcal{N}(F;P)$ is a node of $\mathcal{G}(P)$. Then
  $\mathcal{N}(F;P)$ is also a node of the Gaussian map
  $\mathcal{G}(P_S)$ of a partial sum if and only if $I_F\subseteq S$.
\end{lemma}
\begin{proof}
  Let $F_1+\cdots+F_r$ be the decomposition of $F$, with $\dim(F_i)>0$
  if and only if $i\in I_F$. Since the summands are in general
  orientations, the decomposition is exact and
  $\dim(F)=d-1=\dim(F_1)+\cdots+\dim(F_r)=\sum_{i\in
    I_F}\dim(F_i)$. The normal region $\mathcal{N}(F;P)$ contains a
  single unit vector $l$; and $\mathcal{N}(F;P)$ is a node of
  $\mathcal{G}(P_S)$ if and only if $\dim(\mathcal{S}(P_S;l))=d-1$.
  Again, the decomposition is exact and
  $\dim(\mathcal{S}(P_S;l))=\sum_{i\in
    S}\dim(\mathcal{S}(P_i;l))=\sum_{i\in S}\dim(F_i)$. Since
  $\dim(F_i)>0$ if and only if $i\in I_F$, and $\sum_{i\in
    I_F}\dim(F_i)=d-1$, the result is obvious.
\end{proof}

\section{Sums of polytopes in dimension 3}\label{sec:3d}
To facilitate the comprehension of the proof of
Theorem~\ref{thm:main}, we present informally in this section the
argument for three dimensions, where it is more readily understood. We
present the full proof for general dimensions in
Section~\ref{sec:wit}. The result for three dimensions has already
been published, though only for vertices~\cite{Fogel09}.

In dimension $3$, the Gaussian map of a polytope is a spherical cell
complex of $S^2$, which can be described as a planar graph embedded in
$S^2$. The normal regions of facets, edges and vertices of the
polytope are nodes, edges and faces of the graph respectively. Note
that the normal regions of edges, edges of the graph, are arcs of
great circles of $S^2$. The Gaussian map of a Minkowski sum is the
\emph{overlay} of the Gaussian maps of the summands.

\begin{wrapfigure}{r}{3.4cm}
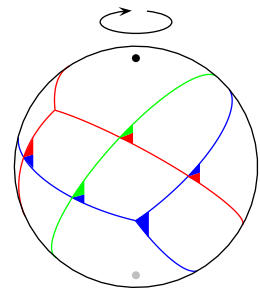

\begin{center}
\psset{unit=1.6cm,shortput=nab,linewidth=0.5pt,arrowsize=2pt 3}
\pspicture(-1,-1)(1,1)%
\psline[linecolor=blue,fillstyle=solid,fillcolor=blue](0,-.446)(0.1,-.38)(0.1,-.58)
\psline[linecolor=red,fillstyle=solid,fillcolor=red](0.43,-.08)(0.52,-.13)(0.52,-.05)
\psline[linecolor=blue,fillstyle=solid,fillcolor=blue](0.43,-.08)(0.52,0.02)(0.52,-.05)
\psline[linecolor=blue,fillstyle=solid,fillcolor=blue](-0.92,0.08)(-0.85,0)(-0.85,0.1)
\psline[linecolor=red,fillstyle=solid,fillcolor=red](-0.92,0.08)(-0.85,0.22)(-0.85,0.1)
\psline[linecolor=red,fillstyle=solid,fillcolor=red](-0.13,0.24)(-0.03,0.2)(-0.03,0.28)
\psline[linecolor=green,fillstyle=solid,fillcolor=green](-0.13,0.24)(-0.03,0.34)(-0.03,0.28)
\psline[linecolor=blue,fillstyle=solid,fillcolor=blue](-0.52,-0.255)(-0.43,-0.3)(-0.43,-0.23)
\psline[linecolor=green,fillstyle=solid,fillcolor=green](-0.52,-0.255)(-0.43,-0.13)(-0.43,-0.23)
\psset{linecolor=blue}
\parametricplot{-39.5}{65}{45 dup 30 dup t cos exch sin mul exch cos t sin mul 3 -1 roll cos mul add exch sin t sin mul}%
\parametricplot{-29}{64}{-30 dup -60 dup t cos exch sin mul exch cos t sin mul 3 -1 roll cos mul add exch sin t sin mul}%
\parametricplot{31}{95}{-60 dup -17 dup t cos exch sin mul exch cos t sin mul 3 -1 roll cos mul add exch sin t sin mul}%
\psset{linecolor=red}
\parametricplot{-30}{35}{55 dup -71 dup t cos exch sin mul exch cos t sin mul 3 -1 roll cos mul add exch sin t sin mul}%
\parametricplot{-33}{-75}{-60 dup -31 dup t cos exch sin mul exch cos t sin mul 3 -1 roll cos mul add exch sin t sin mul}%
\parametricplot{-70}{75}{-30 dup 18 dup t cos exch sin mul exch cos t sin mul 3 -1 roll cos mul add exch sin t sin mul}%
\psset{linecolor=green}
\parametricplot{-80}{100}{50 dup -20 dup t cos exch sin mul exch cos t sin mul 3 -1 roll cos mul add exch sin t sin mul}%
\psset{linecolor=black}
\parametricplot{-180}{180}{90 dup 90 dup t cos exch sin mul exch cos t sin mul 3 -1 roll cos mul add exch sin t sin mul}
\psdots(0,0.9)
\psdots[linecolor=lightgray](0,-0.9)
\psellipticarc{<-}(0,1.2)(0.3,0.1){110}{45}
\endpspicture%
\end{center}
\caption{Example of a Gaussian map, overlay of three Gaussian maps. Each western-most corner exists in the overlay of one or two Gaussian maps
}\label{fig:wmc}
\end{wrapfigure}

Let $P=P_1+\cdots+P_r$ be a sum of $3$-dimensional polytopes in
general orientations. We choose on $S^2$ two antipodal points in
generic position, so that no edge of $\mathcal{G}(P)$ is aligned with
them. In particular, the points are inside two distinct
faces of $\mathcal{G}(P)$. We call these two points \emph{north pole}
and \emph{south pole}. We define \emph{west} in the usual way with
respect to the poles, as a direction turning around the poles,
clockwise from the north pole.

For any spherically convex subset $C$ of $S^2$ that does not contain
either pole, we define the \emph{western-most point} of $C$ as the
point in the closure of $C$ that is to the west of all points in
$C$. We also define as \emph{western-most corner} of $C$ the subset of
$C$ at distance less than $\varepsilon$ of the western-most point, for
a small $\varepsilon>0$. Note that if $C$ is a cell of
$\mathcal{G}(P)$, the western-most point is a node of $\mathcal{G}(P)$
incident to $C$, which is unique because no edge of $\mathcal{G}(P)$
is aligned with the poles.

The normal region $\mathcal{N}(F;P)$ of a facet $F$ of the Minkowski
sum $P$ is a node of $\mathcal{G}(P)$. Because the summands are in
general orientations, there are no three edges of different
$\mathcal{G}(P_i)$ intersecting in a single point. Thus, a node in
$\mathcal{G}(P)$ is either a node in some $\mathcal{G}(P_i)$, in which
case $I_F=\{i\}$, or it is the intersection of two edges in some
$\mathcal{G}(P_i+P_j)$, in which case $I_F=\{i,j\}$ (See
Figure~\ref{fig:wmc}). So a western-most corner in $\mathcal{G}(P)$ is
always a western-most corner in some $\mathcal{G}(P_i)$, or a
western-most corner in some $\mathcal{G}(P_i+P_j)$ whose western-most
point is the intersection of two edges.

So we can find the number of western-most corners in $\mathcal{G}(P)$
by counting those in all $\mathcal{G}(P_i)$ and those in all
$\mathcal{G}(P_i+P_j)$ whose western-most point is an intersection of
edges. But then, the western-most corners in $\mathcal{G}(P_i+P_j)$
also include those whose western-most point is a node of
$\mathcal{G}(P_i)$ or $\mathcal{G}(P_j)$.  Denoting as $w_k(.)$ the
number of western-most corners of $k$-dimensional cells in a Gaussian
map, this means that

$$
w_k(\mathcal{G}(P))=\sum_{i=1}^rw_k(\mathcal{G}(P_i))+ \sum_{1\leq i<j\leq
  r}(w_k(\mathcal{G}(P_i+P_j))-w_k(\mathcal{G}(P_i))-w_k(\mathcal{G}(P_j)))
$$
$$
=\sum_{1\leq i<j\leq r}w_k(\mathcal{G}(P_i+P_j))
-(r-2)\sum_{i=1}^rw_k(\mathcal{G}(P_i)),\quad k=0,1,2.
$$
Intuitively, we sum the number of western-most corners in different
$\mathcal{G}(P_i+P_j)$, and subtract $(r-2)$ times the western-most
corners in the $\mathcal{G}(P_i)$, since they are each counted $(r-1)$
times in the first sum.

But since there is one distinct western-most corner for every cell of
a Gaussian map except the two $2$-dimensional cells that contain a
pole, and the cells of a Gaussian map correspond to faces of the
underlying polytope, we have that for any polytope $\mathcal{P}$,
$w_0(\mathcal{G}(\mathcal{P}))=f_2(\mathcal{P})$,
$w_1(\mathcal{G}(\mathcal{P}))=f_1(\mathcal{P})$ and
$w_2(\mathcal{G}(\mathcal{P}))=f_0(\mathcal{P})-2$. Replacing $w_k$ in
the above equation by these, we get Theorem~\ref{thm:main} for $d=3$.

Here is a subtle detail of the argument. Let us say that a point $p$
in the closure of a subset $C$ of the sphere $S^2$ is a \emph{local
  optimum} of $C$ if $p$ is the western-most point for the
intersection of $C$ with some open set containing $p$. The reason we
use the direction west is that the level curves for west, the
meridians of geography, are arcs of great circles, i.e. geodesics;
they intersect only once any other geodesic inside a spherically
convex set. This guarantees that all local optima are also
western-most points. This would not be the case had we used the
direction south, because the level curves for south, the parallels,
are not geodesics.
\section{Extension to higher dimensions}\label{sec:wit}
We extend in this section the argument of Section~\ref{sec:3d} to
higher dimensions. First, we extend the definition of west and
western-most corners. We prove that the extension has the same
property that (about) every cell of a Gaussian map has a single
western-most corner.
We also prove that the western-most corner of some cell in the
Gaussian map of a Minkowski sum is also a western-most corner of some
cell in any Gaussian map of a partial sum where its western-most point
is a node of the map. Finally, we present the formula that allows us
to count the number of western-most corners in the Gaussian map of the
Minkowski sum.

Here is a brief summary of the proof. Let $P=P_1+\cdots+P_r$ be a
Minkowski sum of $d$-dimensional polytopes in general
orientations. Because the summands are in general orientations, a node
of $\mathcal{G}(P)$ is also a node of $\mathcal{G}(P_S)$ if and only
if the support of its underlying facet is contained in $S$. This
implies that in all $\mathcal{G}(P_S)$ where a node exists, the local
geometry of the map around the node is the same as in
$\mathcal{G}(P)$. Therefore, if the node is a local optimum in
$\mathcal{G}(P)$, it is also a local optimum in any $\mathcal{G}(P_S)$
of which it is a node. This implies that a western-most corner in
$\mathcal{G}(P)$ is also a western-most corner in $\mathcal{G}(P_S)$
if and only if its western-most point in $\mathcal{G}(P)$ exists in
$\mathcal{G}(P_S)$, i.e. if and only if $S$ contains the support of
the underlying facet of the node. This is what allows us to use a counting
argument for deducing the number of western-most corners in
$\mathcal{G}(P)$ from the number of western-most corners in the
Gaussian map of partial sums. Since there is one western-most corner
per cell of the Gaussian map and face of the underlying polytope, this
allows us to find the number of faces of the sum $P$.

We now extend the definition of west. The Gaussian map
$\mathcal{G}(P)$ subdivides the sphere $S^{d-1}$ into a spherical cell
complex. The normal regions of ridges of $P$ are arcs of great circles
on $S^{d-1}$, edges of the Gaussian map $\mathcal{G}(P)$. Each of
these arcs of great circle is contained into the $2$-dimensional
subspace that is orthogonal to the underlying ridge. In the
$d$-dimensional space containing $S^{d-1}$, we choose a linear
subspace $U$ of dimension $d-2$, so that its intersection with every
$2$-dimensional subspace containing an edge of $\mathcal{G}(P)$ is
just the origin.  The next lemma shows that this is always possible.
\begin{lemma}
  In a $d$-dimensional space, for any finite family
  $\{U_1,\ldots,U_n\}$ of $2$-dimensional linear subspaces, it is
  possible to find a linear subspace $U$ of dimension $d-2$ such that
  $U\cap U_i=\{\bvec{0}\}$ for any $i$ in $1,\ldots,n$.
\end{lemma}
\begin{proof}
  The orthogonal complements $U_i^\perp$ of the subspaces $U_i$ in the
  family are of dimension $d-2$. If we choose a vector $\bvec{u}$ that is
  not in these orthogonal complements, then for any $i$,
  $span(\{\bvec{u}\}\cup U_i^\perp)$ is of dimension $d-1$. If we choose now a
  vector $\bvec{v}$ that is not any of these subspaces of dimension
  $d-1$, then for any $i$, $span(\{\bvec{u},\bvec{v}\}\cup U_i^\perp)$
  is the $d$-dimensional space. Define
  $U=span(\{\bvec{u},\bvec{v}\})^\perp$. If a vector is in $U\cap
  U_i$, it is orthogonal to $span(\{\bvec{u},\bvec{v}\}\cup
  U_i^\perp)$, which is the whole space, and so it is the origin
  $\bvec{0}$.
\end{proof}

This is a simple extension of the fact that in three dimensions, for
any number of planes going through the origin, we can choose
a vector that is in none of the planes.

Note that the intersection of $U$ with $S^{d-1}$ is equivalent to the
sphere $S^{d-3}$. If $d=3$, the intersection is the two antipodal
points on $S^2$ that we named north and south poles in
Section~\ref{sec:3d}, and if $d=2$, it is the empty set because
$U=\{\bvec{0}\}$.

We choose an orthonormal basis $\bvec{e_1},\ldots,\bvec{e_d}$ of the
$d$-dimensional space such that
$U=span(\{\bvec{e_3},\ldots,\bvec{e_d}\})$.
We then define successive parametrizations of the spheres
$S^n$, $1\leq n\leq d-1$ as follows:
$$
S^1 = \{\sin(\theta_1)\bvec{e_1}+\cos(\theta_1)\bvec{e_2}\;:\;\theta_1 \in [0,2\pi)\},
$$
$$
S^n = \{\sin(\theta_n)S^{n-1}+\cos(\theta_n)\bvec{e_{n+1}}\;:\;\theta_n \in [0,\pi]\},\quad n=2,\ldots,d-1.
$$
Note that a point of $S^{d-1}$ is in $U$ if and only if
$\sin(\theta_j)=0$ for some $j=2,\ldots,{d-1}$. For any point of
$S^{d-1}$ not in $U$, we define the direction west as
$\dot{\theta}_1$, the direction of augmentation of $\theta_1$. Note
that for $d=3$, it is equivalent to the definition of
Section~\ref{sec:3d}, and for $d=2$, it is a direction running around
$S^1$. In $S^{d-1}$, west is not defined on the subspace $U$ because
$\dot{\theta}_1=0$, so that the intersection of $U$ with $S^{d-1}$ is
a sphere of dimension $d-3$ that plays the same role as poles in three
dimensions.

Recall that in our parametrization of $S^{d-1}$, $\theta_1$ is in
$[0,2\pi)$. Formally, for any points $p$ and $q$ of $S^{d-1}$ that are
not in $U$, we say that $p$ is \emph{to the west} of $q$ if
$\theta_1(p)\in [\theta_1(q),\theta_1(q)+\pi]$ and $\theta_1(q)<\pi$,
or if $\theta_1(p)\in [\theta_1(q),2\pi)\cup[0,\theta_1(q)-\pi]$ and
$\theta_1(q)\geq\pi$.

For any spherically convex subset $C$ of $S^{d-1}$ that does not
intersect $U$, we define the \emph{western-most point} of $C$ as the
point in the closure of $C$ that is to the west of all points in
$C$. The next lemma, proved in Appendix~\ref{app:ls},
shows that the western-most point exists.
\begin{lemma}\label{lem:sense}
  If a spherically convex subset $C$ of $S^{d-1}$ does
  not intersect $U$, it is in a hemisphere defined by $\theta_1\in
  [\alpha,\alpha+\pi]$ or $\theta_1\in [0,\alpha]\cup[\alpha+\pi,2\pi)$ for
  some $\alpha\in[0,\pi)$.
\end{lemma}
We also define as \emph{western-most corner} of $C$ the subset of $C$
at distance less than $\varepsilon$ of the western-most point, where
$\varepsilon>0$ is smaller than the distance between any two
non-incident cells in $\mathcal{G}(P)$. Note that the western-most
point of $C$ is also the western-most point of the western-most corner
of $C$.

Recall that the Gaussian map of a polytope is a subdivision of
$S^{d-1}$ into a spherical cell complex.  For any cell $C$ of
$\mathcal{G}(P)$ that does not intersect $U$, the western-most point
of $C$ is a node incident to $C$, which is unique because otherwise
there would be a great circle containing an edge of $\mathcal{G}(P)$
and intersecting $U$, which contradicts the way we chose $U$. As a
consequence, there is one unique western-most corner for each cell of
a Gaussian map that does not intersect $U$.

We now prove that a western-most corner of some cell of
$\mathcal{G}(P)$ is also a western-most corner of a cell of the
Gaussian map of any partial sum $\mathcal{G}(P_S)$ if and only if its
western-most point is a node of $\mathcal{G}(P_S)$. We call a point
$p$ in the closure of a subset $C$ of $S^{d-1}$ a \emph{local optimum}
of $C$ if $p$ is the western-most point of the intersection of $C$
with some open subset of $S^{d-1}$ containing $p$.
\begin{lemma}\label{lem:local}
  A point $p$ is a local optimum of a cell $C$ of a Gaussian map $G$
  if and only if it is a western-most point of $C$.
\end{lemma}
The proof is in Appendix~\ref{app:ls}. For the next lemma, recall that
$\varepsilon>0$ is smaller than the distance between any two
non-incident cells in $\mathcal{G}(P)$.
\begin{lemma}\label{lem:var}
  Let $F$ be a facet of $P$, with its normal region $\mathcal{N}(F;P)$
  a node of $\mathcal{G}(P)$. Let $p$ be a point of $S^{d-1}$ at
  distance less than $\varepsilon$ of $\mathcal{N}(F;P)$. For any
  partial sum $P_S$ with $I_F\subseteq S$, the dimensions of the cells
  containing $p$ in $\mathcal{G}(P)$ and $\mathcal{G}(P_S)$ are the
  same.
\begin{proof}
  In the Gaussian map $\mathcal{G}(P)$, the subset of $S^{d-1}$ at a
  distance less than $\varepsilon$ of $\mathcal{N}(F;P)$ intersects
  only the normal regions of subfaces of $F$. Therefore, for any point
  $p$ in that subset, $\mathcal{S}(P;p)$ is a subface $G$ of
  $F$. Recall that for any subface $G$ of a facet $F$, $I_G\subseteq
  I_F$. So for any partial sum $P_S$ such that $I_F\subseteq S$,
  $I_G\subseteq S$, which means that not only $P_S$ has a facet with
  the same normal region as $F$, but $\mathcal{S}(P_S;p)$ is a subface
  of that facet with the same dimension as $G$, and $p$ is in a cell
  of the same dimension in $\mathcal{G}(P_S)$ as in $\mathcal{G}(P)$.
\end{proof}
\end{lemma}

We finally have the tools to prove:
\begin{lemma}
  Let $W$ be a western-most corner of a cell $C$ in $\mathcal{G}(P)$,
  with $\mathcal{N}(F;P)$ the western-most point of $C$. Then
  $W$ is a western-most corner of some cell of the same dimension
  in the Gaussian map of a partial sum $\mathcal{G}(P_S)$ if and only
  if $I_F\subseteq S$.
\begin{proof}
  First, if $I_F\not\subseteq S$, then $\mathcal{N}(F;P)$ is not a
  node of $\mathcal{G}(P_S)$, and so $W$ cannot be a western-most
  corner. Suppose $I_F\subseteq S$; then $\mathcal{N}(F;P)$ is a node
  of $\mathcal{G}(P_S)$. Furthermore, by Lemma~\ref{lem:var}, the
  points in $W$ are in a cell of the same dimension in
  $\mathcal{G}(P_S)$ as in $\mathcal{G}(P)$, and the points in the
  closure of $W$ are in a cell of the same dimension in
  $\mathcal{G}(P_S)$ as in $\mathcal{G}(P)$. As a consequence, since
  $\mathcal{N}(F;P)$ is the western-most point of $W$ in
  $\mathcal{G}(P)$, it is also the western-most point of $W$ in
  $\mathcal{G}(P_S)$. But $W$ is the intersection, of the cell it is in,
  with an open subset, and so $\mathcal{N}(F;P)$ is a local optimum of
  the cell that contains $W$ in $\mathcal{G}(P_S)$. By
  Lemma~\ref{lem:local}, it is also the western-most point of the cell
  that contains $W$ in $\mathcal{G}(P_S)$, and so $W$ is the
  western-most corner of that cell in $\mathcal{G}(P_S)$.
\end{proof}
\end{lemma}
This is the most important lemma. It is the ultimate goal of the
definitions in this section, which is to have a witness of the
existence of a cell, a witness whose presence in the Gaussian maps of
partial sums depends on a simple rule.

However, according to the definitions so far, cells intersecting $U$
do not have a western-most corner. In any cell that intersects $U$, it
is possible to turn around $U$, always going west, much like the way
it is possible to turn around a pole on $S^2$. To deal with this
problem, we consider the restriction of the Gaussian map to $U$. Let
us denote as $S_U$ the intersection of $S^{d-1}$ with $U$. $S_U$ is a
sphere equivalent to $S^{d-3}$, and the restriction of a spherical
cell complex on $S^{d-1}$ to $S_U$ also defines a spherical cell
complex on $S_U$. In fact, the restriction to $S_U$ of the Gaussian
map on $S^{d-1}$ of a $d$-dimensional polytope is the Gaussian map on
$S_U$ of the orthogonal projection of the polytope onto $U$.

Since $S_U$ is a sphere equivalent to $S^{d-3}$, we can define west on
$S_U$ as we have done for $S^{d-1}$ (See Figure~\ref{fig:s3}). For
any cell of $\mathcal{G}(P)$ that intersects $S_U$, we define its
western-most corner as the western-most corner of its intersection
with $S_U$ in the restriction of $\mathcal{G}(P)$ to $S_U$. If $d>5$,
this again does not define a westernmost point for every cell, because
west is not defined on the intersection of $S_U$ with a subspace of
dimension $d-4$; so we restrict the Gaussian map to that subspace, and
start again recursively.

\begin{figure}
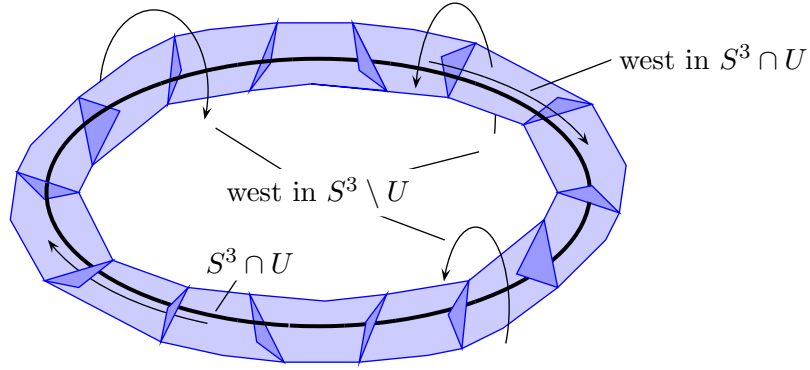

\begin{center}
\psset{unit=0.9cm,shortput=nab,linewidth=0.5pt,arrowsize=2pt 3}
\pspicture(-5,-2.5)(5,2.5)%
\definecolor{verylightblue}{rgb}{0.8,0.8,1}
\definecolor{lightblue}{rgb}{0.6,0.6,1}
\psellipticarc(2,1.2)(0.6,1.6){-30}{50}
\psellipticarc(-2.4,1.5)(0.8,1.2){140}{180}

\psset{fillstyle=solid,fillcolor=verylightblue,linecolor=blue}
\psline(0.5,2.5)(1.4,2.4)(2.3,2.2)(4,1.3)(4.5,0.4)(4.4,-0.3)(4.2,-0.7)(3.5,-1.4)(2.7,-2.0)(2.1,-2.3)(1.6,-2.4)(0.6,-2.5)(-0.5,-2.5)(-1.4,-2.4)(-2.3,-2.2)(-4,-1.3)(-4.5,-0.4)(-4.4,0.3)(-4.2,0.7)(-3.5,1.4)(-2.7,2.0)(-2.1,2.3)(-1.6,2.4)(-0.6,2.5)(0.5,2.5)(1.0,1.5)(-0.1,1.6)(-1,1.5)(-2.2,1.3)(-3.3,0.4)(-3.5,0)(-3.0,-1)(-1.9,-1.4)(-1.0,-1.5)(0.1,-1.6)(1,-1.5)(2.2,-1.3)(3.3,-0.4)(3.5,0)(3.0,1)(1.9,1.4)(-0.1,1.6)
\psset{fillstyle=solid,fillcolor=lightblue}
\psellipticarc[linewidth=1.5pt,fillstyle=none,linecolor=black](0,0)(4,2){61}{80}
\pspolygon(0.5,2.5)(1.0,1.5)(0.6,1.8)
\pspolygon(1.9,1.4)(2.3,2.2)(1.8,2.0)
\psellipticarc[linewidth=1.5pt,fillstyle=none,linecolor=black](0,0)(4,2){35}{61}
\psellipticarc[fillstyle=none,linecolor=black](0,0)(4.2,2.1){35}{50}
\pspolygon(3.0,1)(3.5,1.4)(4,1.3)
\psellipticarc[linewidth=1.5pt,fillstyle=none,linecolor=black](0,0)(4,2){-2}{35}
\psellipticarc[fillstyle=none,linecolor=black]{<-}(0,0)(4.2,2.1){10}{35}
\pspolygon(4.4,-0.3)(3.5,0)(4.0,0.1)
\psellipticarc[linewidth=1.5pt,fillstyle=none,linecolor=black](0,0)(4,2){-35}{-2}
\pspolygon(3.3,-0.4)(3.5,-1.4)(2.9,-1.2)
\psellipticarc[linewidth=1.5pt,fillstyle=none,linecolor=black](0,0)(4,2){-59}{-35}
\pspolygon(2.1,-2.3)(2.2,-1.3)(1.9,-1.8)
\psellipticarc[linewidth=1.5pt,fillstyle=none,linecolor=black](0,0)(4,2){-79}{-59}
\pspolygon(1,-1.5)(0.6,-2.5)(0.7,-1.9)
\psellipticarc[linewidth=1.5pt,fillstyle=none,linecolor=black](0,0)(4,2){-102}{-79}

\psellipticarc[linewidth=1.5pt,fillstyle=none,linecolor=black](0,0)(4,2){80}{101}
\pspolygon(-1,1.5)(-0.6,2.5)(-0.7,1.9)
\psellipticarc[linewidth=1.5pt,fillstyle=none,linecolor=black](0,0)(4,2){101}{122}
\pspolygon(-2.1,2.3)(-2.2,1.3)(-2.0,1.8)
\psellipticarc[linewidth=1.5pt,fillstyle=none,linecolor=black](0,0)(4,2){122}{145}
\pspolygon(-3.3,0.4)(-3.5,1.4)(-2.9,1.2)
\psellipticarc[linewidth=1.5pt,fillstyle=none,linecolor=black](0,0)(4,2){145}{178}
\pspolygon(-4.4,0.3)(-3.5,0)(-4.0,-0.1)
\psellipticarc[linewidth=1.5pt,fillstyle=none,linecolor=black](0,0)(4,2){178}{215}
\psellipticarc[fillstyle=none,linecolor=black]{<-}(0,0)(4.2,2.1){190}{215}
\pspolygon(-3.0,-1)(-3.5,-1.4)(-4,-1.3)
\psellipticarc[fillstyle=none,linecolor=black](0,0)(4.2,2.1){215}{230}

\psellipticarc[linewidth=1.5pt,fillstyle=none,linecolor=black](0,0)(4,2){-145}{-122}
\pspolygon(-0.5,-2.5)(-1.0,-1.5)(-0.9,-2.2)
\pspolygon(-1.9,-1.4)(-2.3,-2.2)(-2.2,-1.7)
\psellipticarc[linewidth=1.5pt,fillstyle=none,linecolor=black](0,0)(4,2){-122}{-102.5}
\psset{fillstyle=none,fillcolor=white,linecolor=black}

\psellipticarc{->}(2,1.2)(0.6,1.6){50}{150}
\psellipticarc{->}(2.3,-1.7)(0.5,1.2){-50}{140}
\psellipticarc{<-}(-2.4,1.5)(0.8,1.2){-30}{140}
\psline(-1,-1)(-1.5,-1.75)
\rput*(-1,-1){$S^3\cap U$}
\psline(4.4,2)(3.1,1.5)
\rput*[l](4.4,2){west in $S^3\cap U$}
\psline(0,0)(-1.5,0.95)
\psline(0,0)(1.85,-0.7)
\psline(0,0)(2.35,0.6)
\rput*(0,0){west in $S^3\setminus U$}
\endpspicture%
\end{center}
\caption{Representation of a map in $S^3$ by stereographic projection
  in Euclidean space. West defined in $S^3\setminus U$ is turning
  around the intersection of $S^3$ with a $2$-dimensional subspace
  $U$.  Cells which intersect the subspace have their western-most
  corner defined by a different direction west defined on the
  intersection, which is equivalent to $S^1$.}\label{fig:s3}
\end{figure}

We present now the complete construction. We have chosen a subspace
$U$ of dimension $d-2$ such that its intersection with any
two-dimensional plane containing an edge of $\mathcal{G}(P)$ is just
the origin. Let us write $U^{d-2}=U$, and denote as $G^{d-2}$ the
restriction of the Gaussian map $\mathcal{G}(P)$ to
$U^{d-2}$. $G^{d-2}$ is a spherical cell complex on $S^{d-3}$. Then,
for any $i$ larger than $2$, we define from $U^i$ and $G^i$ a subspace
$U^{i-2}$, which is a subspace of $U^i$, such that the intersection of
$U^{i-2}$ with any two-dimensional plane containing an edge of $G^i$
is just the origin. We then define $G^{i-2}$ as the restriction of
$G^i$ to $U^{i-2}$, which is a spherical cell complex on
$S^{i-3}$. This defines a sequence of subspaces $U^{d-2}\supset
U^{d-4}\supset\cdots$ and a sequence of spherical cell complexes
$G^{d-2}\supset G^{d-4}\supset\cdots$. If $d$ is even, the sequences
end with $G^2$, which is a spherical cell complex on $S^1$, and $U^0$
is just the origin and does not intersect $S^1$. If $d$ is odd, they
end with $G^3$, a spherical cell complex on $S^2$, and $U^1$ is a
one-dimensional subspace whose intersection with $S^2$ defines two
antipodal points that we called north and south pole in
Section~\ref{sec:3d}.

For each $G^i$, spherical cell complex of $S^{i-1}$, we can define a
direction west for every point of $S^{i-1}$ that is not on $U^{i-2}$,
as we have done for $\mathcal{G}(P)$ on $S^{d-1}$. Then for any cell
$C$ of $\mathcal{G}(P)$, let $i$ be the smallest number such that the
intersection of $C$ with $U^i$ is nonempty. We then define the
western-most point of $C$ to be the western-most point of $C\cap U^i$,
cell of $G^i$ on the sphere $S^{i-1}$. The western-most corner of $C$
is also the western-most corner of $C\cap U^i$.

Note that if a cell $C$ is of dimension $d-k-1$, i.e. it is the normal
region in $\mathcal{G}(P)$ of a $k$-dimensional face, it does not
intersect $U^i$ for any $i\geq k$, because $U^i$ was chosen so as not
to intersect edges and nodes of $G^{i+2}$, which are restrictions to
$G^{i+2}$ of cells of dimension $d-i-1$ and $d-i-2$ in
$\mathcal{G}(P)$. For instance, if $d$ is odd, only normal region of
vertices may intersect with $U^1$. Since $U^1$ only intersects
$S^{d-1}$ in two antipodal points, there are exactly two cells of
dimension $d-1$ in any Gaussian map that intersect $U^1$. These are
the only two cells that do not have a western-most corner. When $d$ is
even, west is defined on every point of $G^2$, spherical cell complex
on $S^1$, and so every cell of a Gaussian map has a western-most
corner.

We have now defined a western-most corner for every cell of Gaussian
maps, with the exception, if $d$ is odd, of the two cells that contain
a pole. As before for cells that do not intersect $U$, the
western-most corner of a cell of $\mathcal{G}(P)$ is also a
western-most corner of a cell of the same dimension in the Gaussian
map of a partial sum $\mathcal{G}(P_S)$ if and only if $S$ contains
the support of its western-most point, or rather the support of the
cell whose restriction is its western-most point. The cardinality of
the support is always less than $d$.

Now that we have a complete definition of western-most corners, all
that remains is to count them. The support of any face of $P$ has
cardinality less than $d$, so all western-most corners of cells of
$\mathcal{G}(P)$ can be found in the Gaussian map of partial sums of
at most $d-1$ summands. It is not difficult to see that for any $j\geq
|I_F|$, there are ${{r-|I_F|} \choose {j-|I_F|}}$ subsets of
$\{1,\ldots,r\}$ of cardinality $j$ that contain $I_F$.

The formula of the main theorem was found by observing low-dimensional
cases. It is based on the following combinatorial equivalence:
\begin{lemma}\label{lem:for}
For any $1\leq s<d\leq r$,
$$
\sum_{j=1}^{d-1}(-1)^{d-1-j} {{r-1-j} \choose {d-1-j}} {{r-s} \choose
  {j-s}} = 1.
$$
\end{lemma}
The proof is in Appendix~\ref{app:for}.

By this Lemma, if we count all the western-most corners in partial
sums of $j$ polytopes, multiply by $(-1)^{d-1-j}{r-1-j \choose
  d-1-j}$, and sum over $j$, we end up counting exactly once each
western-most corner, no matter what is the cardinality of the relevant
support. Therefore, if $w_k(.)$ is the number of western-most corners
of $k$-dimensional cells in a Gaussian map,
$$
w_k(\mathcal{G}(P)) = \sum_{j=1}^{d-1}(-1)^{d-1-j} {{r-1-j} \choose
  {d-1-j}} \sum_{S\in \mathcal{C}_j^r}w_k(\mathcal{G}(P_S)),\quad k=0,\ldots,d-1
$$
where $\mathcal{C}_j^r$ is the family of subsets of $\{1,\ldots,r\}$
of cardinality $j$. Since there is one western-most corner of a
$k$-dimensional cell for each $d-1-k$ face of the underlying polytope,
this proves Theorem~\ref{thm:main} for any $d$ and $k$. The only
exception is that if $d$ is odd, any Gaussian map has two regions of
dimension $d-1$ that contain the poles, and that have no western-most
corner, and so in that case, $w_{d-1}(\mathcal{G}(\mathcal{P}))=
f_0(\mathcal{P})-2$. This gives the special case of the theorem for
$d$ odd and $k=0$.

In order to prove Corollary~\ref{cor}, it is enough to point out that
each western-most corner of cells of $\mathcal{G}(P)$ can be found at
least once (and often a lot more) in the Gaussian map of partial sums
of $d-1$ summands. And so, the last term of the sum in
Theorem~\ref{thm:main} is an upper bound on the number of faces.


\section{Maximum number of vertices}\label{sec:max}
Using Corollary~\ref{cor}, we show bounds on the number of vertices of
Minkowski sums. The trivial bound tells us that if $r<d$, then
$f_0(P_1+\cdots+P_r)\leq \prod_{i=1}^rf_0(P_i)$. Consequently, if
$r\geq d$, we get by Corollary~\ref{cor} that
$$
f_0(P_1+\cdots+P_r) \leq \sum_{S\in \mathcal{C}_{d-1}^r}\prod_{i\in S}f_0(P_i).
$$
This can be seen as enumerating all possible combinations of $d-1$
vertices chosen each from a different summand. This is necessarily
lower than all possible combinations of $d-1$ vertices from the
summands. If the summands have $n$ vertices in total, this upper bound
is ${n \choose d-1}$, which is in $O(n^{d-1})$. If each summands has
$n$ vertices, then we have:
$$
f_0(P_1+\cdots+P_r) \leq \sum_{S\in
  \mathcal{C}_{d-1}^r}n^{d-1}={r \choose d-1}n^{d-1},
$$
which is in $O(r^{d-1}n^{d-1})$. The previous known bound was in
$O(r^{d-1}n^{2(d-1)})$~\cite{Gritzmann93}.

Note that a construction from~\cite{Fukuda07} allows us to choose
$d-1$ polytopes of $n$ vertices each such that the sum has $n^{d-1}$
vertices. It is easy to adapt this construction to choose $r$
polytopes such that any partial sum of $d-1$ summands has this many
vertices, which by Theorem~\ref{thm:main} means that the total sum has
exactly (for $d$ even) $\sum_{j=1}^{d-1}(-1)^{d-1-j}{r-1-j\choose
  d-1-j}{r\choose j}n^j$ vertices.

Unfortunately, except in three dimensions, the maximum number of
facets of a Minkowski sum of polytopes remains open, even for two
summands in four dimensions, so we cannot write an upper bound for facets.
We can however tell that if the number of facets in the sum of $d-1$ polytopes
is in $O(p(n))$, their number in the sum of $r\geq d$ polytopes is in
$O(r^{d-1}p(n))$. Finding $p(n)$ should be the object of further research.

\section{Summary}\label{sec:con}
We have extended the intuitive concept of west from three dimensions
to higher dimensions. Thanks to the properties of the concept, we were
able to prove a relation on the number of vertices in sums of many
polytopes, and show that this number has a comparatively low order
of complexity.  For faces of higher dimensions, the result also shows
that the complexity of Minkowski sums of many polytopes is not much
more complex than that of $d-1$ summands.

\appendix
\section{Appendix}
\subsection{Proof of Lemma~\ref{lem:sense} and \ref{lem:local}}\label{app:ls}

We prove in this appendix that the subspace $U$ plays the same role as
poles in $3$ dimensions, and that our definition of west has the
property that if we ``optimize'' in direction west over a spherically
convex subset of $S^{d-1}$, a local optimum of the subset is also a
global optimum. We start with a few lemmas:
\begin{lemma}\label{lem:subsp}
  Let $p$ and $p'$ be distinct non-antipodal points of $S^{d-1}$.
  Suppose $p$ and $p'$ are on a same subspace. Then all points on the
  great circle of $S^{d-1}$ containing $p$ and $p'$ are on that subspace.
  \begin{proof}
    A great circle of $S^{d-1}$ is the intersection of $S^{d-1}$ with
    a $2$-dimensional space. Suppose a great circle contains $p$, $p'$
    and $q$, with $p$ and $p'$ on a subspace $L$, but $q\not\in L$.
    Then the intersection of the $2$-dimensional space containing the
    great circle with $L$ is $1$-dimensional, and so it is a line
    going through the origin. Since $p$ and $p'$ are both on that line
    and in $S^{d-1}$, they are either the same or antipodal.
  \end{proof}
\end{lemma}

For any $\theta$, let us denote as $L(\theta)$ the subspace
orthogonal to $\sin(\theta)\bvec{e_1}+ \cos(\theta)\bvec{e_2}$,
i.e. the set $\{p\;:\;\langle
p,\sin(\theta)\bvec{e_1}+\cos(\theta)\bvec{e_2}\rangle=0\}$.
\begin{lemma}\label{lem:orth}
  A point $p$ of $S^{d-1}$ is in $L(\theta)\cap S^{d-1}$ if and only
  if $\cos(\theta_1(p)-\theta))=0$ or $p\in U$.
\begin{proof}
  Any $p$ in $S^{d-1}$ is written in our parametrization as $\rho
  s+u$, with $\rho\geq 0$, $s\in S^1$ and $u\in U$. We can write
  $s=\sin(\theta_1(p))\bvec{e_1}+\cos(\theta_1(p))\bvec{e_2}$, and
  $\rho=0$ if and only if $p\in U$. Therefore, $\langle
  p,\sin(\theta)\bvec{e_1}+\cos(\theta)\bvec{e_2}\rangle=
  \rho(\sin(\theta)\sin(\theta_1(p))+\cos(\theta)\cos(\theta_1(p)))=
  \rho\cos(\theta_1(p)-\theta)$. So $p$ is in $L$ if and only if $\rho
  \cos(\theta_1(p)-\theta)=0$, which is if and only if $\rho=0$ or
  $\cos(\theta_1(p)-\theta)=0$. The result follows.
\end{proof}
\end{lemma}

\begin{lemma}\label{lem:two}
  Let $K$ be great circle of $S^{d-1}$. Then either $K$ is inside $U$;
  or $K$ intersects $U$, $K\setminus U$ has two connected components
  $K_1$ and $K_2$ such that for any two points $p\in K_1$, $p'\in
  K_2$, $\theta_1(p)+\pi=\theta_1(p')$; or $K$ does not intersect $U$,
  and for any two distinct points $p$, $p'$ in $K$, $\theta_1(p)\ne
  \theta_1(p')$, and $\theta_1(p)+\pi=\theta_1(p')$ if and only if
  $p$ and $p'$ are antipodal.
\begin{proof}
  Suppose $p$ and $p'$ distinct in $K$ such that
  $\theta_1(p)=\theta_1(p')$; or suppose $p$ and $p'$ non-antipodal in
  $K$ such that $\theta_1(p)+\pi=\theta_1(p')$; or suppose $p'$ is in
  $K\cap U$ and any $p$ in $K$. In all three cases, by
  Lemma~\ref{lem:orth}, $p$ and $p'$ are both in the subspace
  $L(\theta_1(p)+\pi/2)$. By Lemma~\ref{lem:subsp}, all points on $K$
  are in $L(\theta_1(p)+\pi/2)$. Then by Lemma~\ref{lem:orth} again,
  for any $q$ on the arc of great circle, $q\in U$ or
  $\cos(\theta_1(q)-(\theta_1(p)+\pi/2))=0$. Suppose $K\cap U$ contains
  more than two points. Then some of them are distinct and
  non-antipodal, and by Lemma~\ref{lem:subsp}, $K$ in inside $U$.
  Otherwise, for any $q$ and $q'$ antipodal on $K\setminus U$,
  $\theta_1(q')=\theta_1(q)\pm \pi$. So there must be two antipodal
  points of $K$ inside $U$ separating $q$ and $q'$, and so $K\setminus
  U$ has two connected components.

  The only remaining case is that for any $p$ and $p'$ distinct in
  $K$, $\theta_1(p)\ne\theta_1(p')$; for any $p$ and $p'$ in $K$ such
  that $\theta_1(p)+\pi=\theta_1(p')$, $p$ and $p'$ must be antipodal; and
  $K\cap U$ is empty.
\end{proof}
\end{lemma}
Note that if a great circle of $S^{d-1}$ does not intersect $U$, then
$\theta_1$ is different in any two points of the great circle. Since
the parametrization is smooth on $S^{d-1}\setminus U$, $\theta_1$
augments monotonically and continuously in one direction around the
great circle, except in one point when it drops from $2\pi$ to $0$.

Let us recall Lemma~\ref{lem:sense} before proving it:
\paragraph{Lemma~\ref{lem:sense}}
{\em
  If a spherically convex subset $C$ of $S^{d-1}$ does
  not intersect $U$, it is in a hemisphere defined by $\theta_1\in
  [\alpha,\alpha+\pi]$ or $\theta_1\in [0,\alpha]\cup[\alpha+\pi,2\pi)]$ for
  some $\alpha\in[0,2\pi)$.
}
\begin{proof}
  Let $T_1(C)$ be the set of values of $\theta_1$ over $C$ in our
  parametrization. Since $C$ does not intersect $U$, $T_1(C)$ is
  connected. Suppose $T_1(C)$ is $[0;2\pi)$, then there are two points
  $p$, $p'$ in $C$ with $\theta_1(p)+\pi=\theta_1(p')$. Because $C$ is
  spherically convex, any shortest arc of great circle between $p$ and
  $p'$ is contained in $C$. If $p$ and $p'$ are antipodal, then $C$ is
  the whole sphere and intersects $U$. Otherwise, by
  Lemma~\ref{lem:two}, the arc of great circle again contains a point
  in $U$. This is a contradiction.
  
  Otherwise, suppose without loss of generality that the supremum
  of $T_1(C)$ is $3\pi/2$. Then either $C$ is in the hemisphere
  defined by $\theta_1\in[\pi/2,3\pi/2]$, or there is a $\delta>0$
  such that there are two points $p$, $p'$ in $C$ with
  $\theta_1(p)+\pi=\theta_1(p')=3/2-\delta$. As above, this implies
  that $C$ contains a point in $U$, which is a contradiction.
\end{proof}

Let us recall Lemma~\ref{lem:local} before proving it:
\paragraph{Lemma~\ref{lem:local}}
{\em
  A point $p$ is a local optimum of a cell $C$ of a Gaussian map $G$
  if and only if it is a western-most point of $C$.
}
\begin{proof}
  By definition, a western-most point is always a local
  optimum. Assume $p$ is a local optimum of $C$, and that some
  distinct $p'$ is the western-most point of $C$, and therefore also a
  local optimum. Then the shortest arc of great circle between $p$ and
  $p'$ is in $C$. Let $\alpha=\theta_1(p')-\theta_1(p)$.  If
  $\cos(\alpha)\ne 0$, then by Lemma~\ref{lem:two}, the great circle
  defined by $p$ and $p'$ does not intersect $U$, and $\theta_1$
  augments continuously from $p$ to $p'$ except possibly in one point
  when it jumps from $2\pi$ to $0$. So there is a $q$ in the
  intersection of the arc of great circle from $p$ to $p'$ with the
  open set that proves $p$ is a local optimum. For $q$ close enough,
  $\theta_1(q)>\theta_1(p)$, and so $p$ is not a local optimum, a
  contradiction.

  Suppose now $\alpha=\pi$. If $p$ and $p'$ are antipodal, then any
  great circle containing $p$ and $p'$ is in $C$, and $C$ is the whole
  sphere, a contradiction. If $p$ and $p'$ are not antipodal, then by
  Lemma~\ref{lem:two}, there is a point in the arc of great circle
  from $p$ to $p'$ that is in $U$, and so $C$ intersects $U$. But this
  means $C$ has no western-most point, a contradiction.

  Suppose now $\alpha=0$. Then $p$ and $p'$ are incident to a cell
  where $\theta_1$ is fixed. But this means that the great circles
  containing edges of the cell intersect $U$, which contradicts the
  way we have chosen $U$.

  Therefore, it is impossible to have a local optimum $p$ and
  a distinct western-most point $p'$ of a same cell.
\end{proof}

\subsection{Proof of Lemma~\ref{lem:for}}\label{app:for}

We prove here the combinatorial equivalence used for the formulation
of Theorem~\ref{thm:main}. The relation can also be derived from a
protean family of equivalences of type
$$
\sum_{j=0}^c(-1)^{c-j}{a+j\choose b+c}{c\choose j} = {a\choose b},\quad b<a,\;c\geq 0.
$$
See also~\cite[p. 169]{Graham94}, \cite[p. 149]{Grunbaum67}, \cite[p. 285]{Ziegler95} on this subject.

Let us recall Lemma~\ref{lem:for} before proving it:
\paragraph{Lemma~\ref{lem:for}}
{\em For any $1\leq s<d\leq r$,
$$
\sum_{j=1}^{d-1}(-1)^{d-1-j} {{r-1-j} \choose {d-1-j}} {{r-s} \choose
  {j-s}} = 1.
$$
}

\begin{proof}
We prove the Lemma by induction over $r$. We know that for any $d$,
$\sum_{j=0}^d(-1)^j{d \choose j}=0$. We can also write
$\sum_{j=0}^d(-1)^j{d-s \choose j-s}=0$, for any $1\leq s<d$, and so we have
$\sum_{j=0}^{d-1}(-1)^{d-1-j}{d-s \choose j-s}=1$. We can also write
$$
\sum_{j=0}^{d-1}(-1)^{d-1-j}{d-1-j \choose d-1-j}{d-s \choose j-s}=1.
$$
This proves the relation for $r=d$. Assume the relation is proved for
$r$. Then
$$
\sum_{j=0}^{d-1}(-1)^{d-1-j}{r-1-j \choose d-1-j}{r-s \choose j-s}=1,
$$
$$
=\sum_{j=0}^{d-1}(-1)^{d-1-j}{r-j \choose
  d-1-j}{r-s \choose j-s}- \sum_{j=0}^{d-1}(-1)^{d-1-j}{r-1-j
  \choose d-2-j}{r-s \choose j-s}.
$$
Replacing $j$ in the second sum with $j'-1$, we get
$$
\sum_{j=0}^{d-1}(-1)^{d-1-j}{r-j \choose
  d-1-j}{r-s \choose j-s}- \sum_{j'=1}^{d}(-1)^{d-j'}{r-j' \choose
  d-1-j'}{r-s \choose j'-1-s}=1.
$$
In the second sum, the term $j'=d$ gives zero, so we can remove it and add one
for $j'=0$, which also gives zero.
$$
\sum_{j=0}^{d-1}(-1)^{d-1-j}{r-j \choose d-1-j}{r-s \choose j-s}-
\sum_{j'=0}^{d-1}(-1)^{d-j'}{r-j' \choose d-1-j'}{r-s \choose
  j'-1-s}=1.
$$
Grouping the sums, we get
$$
\sum_{j=0}^{d-1}(-1)^{d-1-j}{(r+1)-1-j \choose d-1-j}{(r+1)-s
  \choose j-s}=1,
$$
and so the relation is true for $r+1$, which proves
Lemma~\ref{lem:for} by induction.
\end{proof}


\begin{thebibliography}{10}

\bibitem{Fogel09}
E.~Fogel, D.~Halperin, and C.~Weibel.
\newblock On the exact maximum complexity of minkowski sums of polytopes.
\newblock {\em Discrete and Computational Geometry}, 42(4):654--669, 2009.

\bibitem{Fukuda07}
K.~Fukuda and C.~Weibel.
\newblock On f-vectors of {M}inkowski additions of convex polytopes.
\newblock {\em Discrete and {C}omputational {G}eometry}, 37:503--516, 2007.

\bibitem{Fukuda08}
K.~Fukuda and C.~Weibel.
\newblock Minkowski sums of polytopes relatively in general position.
\newblock Accepted by {\em European {J}ournal of {C}ombinatorics}, 2008.

\bibitem{Graham94}
R.~L. Graham, D.~E. Knuth, and O.~Patashnik.
\newblock {\em Concrete Mathematics: A Foundation for Computer Science}.
\newblock Addison-Wesley Longman Publishing Co., Inc., Boston, MA, USA, 1994.

\bibitem{Gritzmann93}
P.~Gritzmann and B.~Sturmfels.
\newblock Minkowski addition of polytopes: computational complexity and
  applications to {G}r\"obner bases.
\newblock {\em SIAM Journal on Discrete Mathematics}, 6(2):246--269, 1993.

\bibitem{Grunbaum67}
B.~Gr{\"u}nbaum.
\newblock {\em Convex polytopes}.
\newblock With the cooperation of Victor Klee, M. A. Perles and G. C. Shephard.
  Pure and Applied Mathematics, Vol. 16. Interscience Publishers John Wiley \&
  Sons, Inc., New York, 1967.

\bibitem{Halperin04}
D.~Halperin, L.~E. Kavraki, and J.-C. Latombe.
\newblock Robotics.
\newblock In J.~E. Goodman and J.~O'Rourke, editors, {\em Handbook of Discrete
  and Computational Geometry}, chapter~48, pages 1065--1093. CRC Press LLC,
  Boca Raton, FL, 2004.

\bibitem{Lozano79}
T.~Lozano-P\'erez and M.~A. Wesley.
\newblock An algorithm for planning collision-free paths among polyhedral
  obstacles.
\newblock {\em Communications of the ACM}, 22(10):560--570, 1979.

\bibitem{Pachter05}
L.~Pachter and B.~Sturmfels, editors.
\newblock {\em Algebraic statistics for computational biology}.
\newblock Cambridge University Press, New York, 2005.

\bibitem{Sanyal07}
R.~Sanyal.
\newblock Topological obstructions for vertex numbers of {M}inkowski sums.
\newblock preprint, 2007.

\bibitem{Weibel07}
C.~Weibel.
\newblock {\em Minkowski sums of polytopes: Combinatorics and Computation}.
\newblock PhD thesis, EPFL, Lausanne, 2007.

\bibitem{Zhang09}
H.~Zhang.
\newblock {Partially Observable Markov Decision Processes: A Geometric
  Technique and Analysis}.
\newblock {\em Operations Research}, doi:10.1287/opre.1090.0697, 2009.

\bibitem{Ziegler95}
G.~M. Ziegler.
\newblock {\em Lectures on polytopes}, volume 152 of {\em Graduate Texts in
  Mathematics}.
\newblock Springer-Verlag, New York, 1995.

\end{thebibliography}
\end{document}